\documentclass[11pt]{article}

\usepackage[a4paper,margin=1in]{geometry}
\usepackage[utf8]{inputenc}
\usepackage[T1]{fontenc}
\usepackage{lmodern}
\usepackage{microtype}
\usepackage{amsmath,amssymb,amsfonts,amsthm,mathtools}
\usepackage{bm}
\usepackage{booktabs}
\usepackage{array}
\usepackage{enumitem}
\usepackage[hidelinks]{hyperref}
\usepackage{tikz}
\usetikzlibrary{arrows.meta,positioning,automata,shapes.geometric}
\newtheorem{theorem}{Theorem}
\newtheorem{lemma}{Lemma}
\newtheorem{proposition}{Proposition}
\newtheorem{corollary}{Corollary}
\theoremstyle{definition}
\newtheorem{definition}{Definition}

\newtheorem{example}{Example}
\newtheorem{procedure}{Procedure}
\theoremstyle{remark}

\DeclareMathOperator{\cone}{cone}

\DeclareMathOperator{\sign}{sign}

\title{\bfseries Cone-Induced Observation Congruences\\
for Vector-Valued Quantitative Languages}

\author{
Faruk Alpay$^{1,*}$ \and Baris Başaran$^{1}$\\[2mm]
{\small $^{1}$Department of Computer Engineering, Bahcesehir University, Istanbul, Turkey}\\
{\small faruk.alpay@bahcesehir.edu.tr; baris.basaran@bahcesehir.edu.tr}\\
{\small $^{*}$Correspondence: alpay@lightcap.ai.}
}
\date{}

\begin{document}
\maketitle

\begin{abstract}
We study the observation congruences induced by rational polyhedral cones on vector-valued quantitative languages. The extreme rays of the dual cone define intrinsic covectors, and these covectors classify every incremental residual future by a finite sign cell: negative, tight, or positive along each extremal Farkas direction. The resulting carrier is the right-stable carrier of this cone-induced observation family, whose source is canonical: the restricted covector geometry of the order cone on the residual span of the language. We organize this construction through an observation-refinement correspondence, a cone-refinement calculus, and a separation between the qualitative conic observation quotient and the numerical residual carrier needed for potential certificates. A bounded-horizon fragment is fully computable by enumeration of accumulated futures, and reproducible evaluation runs show how the conic layer detects qualitative obstruction cells before numerical refinement.
\end{abstract}


\section{Introduction}

A vector-valued quantitative language can be minimized in several different ways. The ordinary residual quotient keeps equality of all vector futures. A numerical verification procedure may instead keep scalar values needed by potentials, shortest paths, or terminal offsets. A cone order suggests a third layer: it asks which extremal dual witnesses of the cone see a residual future as negative, tight, or positive. This paper studies the right congruence generated by that cone-induced covector observation.

Let
\[
        h:\Sigma^\ast\to\mathbb Q^d
\]
be a vector-valued quantitative language. For a prefix $u$, its incremental future is
\[
        \partial_u h(z)=h(uz)-h(u),\qquad z\in\Sigma^\ast .
\]
If $K$ is a rational polyhedral comparison cone, the extreme rays of $K^\ast$ define a central hyperplane arrangement
\[
        r(x)=0,\qquad r\in\mathcal E(K^\ast),
\]
and a covector map
\[
        \kappa_K(x)=\bigl(\operatorname{sign}(r(x))\bigr)_{r\in\mathcal E(K^\ast)}.
\]
Thus the residual orbit of $h$ is stratified by extremal Farkas sign cells. The carrier studied here is the right-action carrier of this cone-induced stratification under $\Sigma^\ast$.

The mathematical organization is the cone-to-observation map: the cone supplies a canonical covector family; restrictions of this family to the residual span produce a presentation-independent sign configuration; refinements of the dual covector arrangement induce quotient morphisms; and the resulting qualitative carrier is separated from the magnitude-sensitive numerical carrier used by exact verification algorithms.

This distinction is important in verification. A negative conic cell gives an extremal Farkas witness that a residual future leaves the cone. A positive conic cell does not determine the size of the residual or the feasibility of a shortest-path potential. The conic carrier is therefore useful as a structural abstraction and, in bounded horizons, as a computable pre-analysis. Exact numerical inclusion remains a refinement layer.

The paper is deliberately finite-word and deterministic in its main body. The unbounded construction is a characterization unless a future-separation procedure is available, while the bounded-horizon fragment is executable by enumerating accumulated residual futures. Conditional spectral material is treated as a supplementary note.

\subsection{Running example: a cone-induced covector split}

Consider the wedge cone
\[
        K_\triangle=\{(x_1,x_2):x_2-x_1\ge0,\ x_2+x_1\ge0\},
\]
whose two extreme dual rays may be represented by
\[
        r^- =(-1,1),\qquad r^+=(1,1).
\]
The covector of a residual vector $x$ is
\[
        \kappa_{K_\triangle}(x)=\bigl(\sign(r^-(x)),\sign(r^+(x))\bigr).
\]
The following two-state carrier illustrates the distinction between value equality, conic cell information, and numerical verification data:

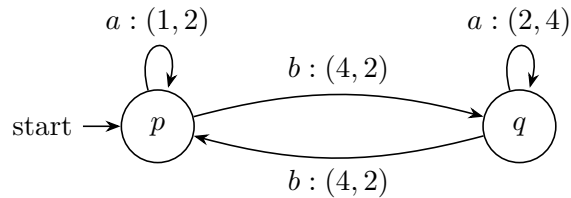
\begin{center}
\begin{tikzpicture}[->,>=Stealth,node distance=38mm,semithick]
  \node[state,initial] (p) {$p$};
  \node[state,right=of p] (q) {$q$};
  \path (p) edge[loop above] node {$a:(1,2)$} (p);
  \path (q) edge[loop above] node {$a:(2,4)$} (q);
  \path (p) edge[bend left=15] node[above] {$b:(4,2)$} (q);
  \path (q) edge[bend left=15] node[below] {$b:(4,2)$} (p);
\end{tikzpicture}
\end{center}

The two $a$-loops have different vector values, but both lie in the open cell $(+,+)$ because
\[
        r^-(1,2)=1,\quad r^+(1,2)=3,
        \qquad
        r^-(2,4)=2,\quad r^+(2,4)=6.
\]
Thus a conic covector carrier may identify these two positive-cell futures even though a numerical residual carrier cannot. The $b$-edge has profile $(-,+)$ because $r^-(4,2)=-2$ and $r^+(4,2)=6$, so it is a qualitative Farkas obstruction. This example also shows the limitation: the conic carrier detects the negative cell on $b$, but it does not decide shortest-path distances among positive-cell behaviours. Numerical potential checks remain a separate layer.

\subsection{A four-block covector carrier}

A slightly larger finite example shows that the quotient need not collapse to one or two states. Again take $K_\triangle$ with rays $r^-=(-1,1)$ and $r^+=(1,1)$. Consider four deterministic states $s_{++},s_{0+},s_{-+},s_{+0}$ with a single self-loop letter $a$ and zero terminal values. The loop labels are
\[
\begin{array}{c|c|c}
\text{state} & \text{loop label} & (\sign r^-(x),\sign r^+(x))\\
\hline
s_{++} & (1,2) & (+,+)\\
s_{0+} & (2,2) & (0,+)\\
s_{-+} & (4,2) & (-,+)\\
s_{+0} & (-2,2) & (+,0).
\end{array}
\]
Since every state is a loop state, the residual future of $a^n$ repeats the same cell for each state. The conic observation carrier therefore has four blocks. This is still a small example, but it demonstrates a nontrivial covector alphabet: the quotient is not merely detecting violation versus non-violation; it distinguishes negative, tight, and positive extremal Farkas behaviour. If the labels $(1,2)$ and $(2,4)$ are both used instead of $(1,2)$ and $(2,2)$, they remain in the same $(+,+)$ cell and are merged by the conic carrier while remaining distinct in the numerical carrier.

\subsection{Contributions}

The paper's contribution is a conic observation calculus rather than a new general minimization theorem. Once an observation family is fixed, the carrier construction is classical; what is specific here is that the observation family is generated by the dual geometry of the comparison cone and that cone-geometric refinements induce systematic operations on congruences.

\begin{enumerate}[label=(C\arabic*)]
\item \textbf{Cone-induced covector observations.} We define the intrinsic covector map $\kappa_K$ induced by the extreme rays of $K^\ast$ and apply it to every incremental residual future of a vector-valued quantitative language.
\item \textbf{Lineality-first reduction.} We show that the conic observation quotient is formed after quotienting by the lineality space $L_K$; order-neutral residual directions are invisible to every extreme dual witness.
\item \textbf{Observation-refinement correspondence.} We prove an order-theoretic correspondence between right congruences and the dual covectors whose residual sign languages they preserve. This identifies exactly how much dual geometry is visible at a given quotient.
\item \textbf{Cone-refinement calculus.} We show that refinement of the extreme-dual sign arrangement refines the conic congruence and induces canonical quotient morphisms. Union of covector families becomes meet of congruences.
\item \textbf{Separation from numerical verification.} We distinguish the conic cell carrier from the numerical residual carrier needed for potential certificates; the former records extremal obstruction cells, while the latter preserves magnitudes.
\item \textbf{Computable bounded-horizon fragment.} We give a finite, executable bounded-horizon construction and a worked run showing how accumulated residual signs, not edge-local signs, determine the quotient. The unbounded construction is presented as a characterization unless a separation procedure is available.
\end{enumerate}

\paragraph{What is new.}
The Myhill--Nerode part of the construction is deliberately not presented as new: after an observation family has been fixed, the coarsest right congruence preserving it is the classical object. The new content is the cone-induced source and organization of the observations. The extreme dual rays of $K$ supply the observation coordinates, lineality passes first to $V/L_K$, residual futures restrict the sign arrangement to the actually reachable span, and cone-refinement becomes a comparison principle for the resulting quotients. Propositions that identify row kernels or factorization properties should be read as structural packaging of this observation family; the cone-specific claims are the lineality reduction, the restricted covector invariance, the observable-covector closure, the refinement morphisms, and the bounded accumulated-future construction.

\section{Related Work and Scope}

Recognizable series and weighted automata provide the algebraic basis for quantitative languages \cite{schutzenberger1961,eilenberg1974,berstel2011,droste2009}. Weighted MSO and weighted automata equivalence provide logical denotations \cite{droste2007,droste2019}, while multi-weighted automata extend the denotational layer to several weights \cite{droste2016multi}. Quantitative languages study numerical comparison, inclusion, and equivalence \cite{chatterjee2010}; cost-register automata and regular cost functions study word-to-value computations through registers \cite{alur2013}.

The present construction is related to syntactic congruence theory \cite{eilenberg1974,pin2021}, but its observation family is not arbitrary. It is induced by the extreme-dual covectors of the cone and then restricted to the residual span of the language. Classical syntactic congruences ask which prefixes preserve a chosen set of languages or observations. Here the observations are the extremal Farkas sign cells of residual futures. The resulting carrier can therefore be read as the syntactic carrier of a cone-induced finite covector alphabet.

The work is also related to abstract interpretation and sign abstraction \cite{cousot1977,cousot2021}. The abstraction map is a sign map, but it is not a coordinate sign domain chosen independently of the semantics. Its coordinates are the extreme rays of $K^\ast$, so the abstraction is invariant under positive rescaling and relabeling of those rays. The restricted sign configuration can also be described using the standard language of realizable oriented matroids \cite{bjorner1999}; this is used only as a concise way to say that deletion, restriction, and contraction are presentation-independent operations on the induced sign arrangement. Adding redundant non-extreme inequalities changes a presentation but not the presentation-independent extreme-ray carrier unless additional hyperplanes are intentionally retained as extra observations.

Difference constraints, potentials, negative cycles, and Farkas separation remain the exact scalar certificate tools \cite{bellman1958,ford1956,farkas1902}. Quantitative simulations by matrices provide linear witnesses for weighted inclusion \cite{urabe2017}. The conic carrier does not encode scalar magnitudes and therefore does not replace these mechanisms. It gives a qualitative layer that may expose extremal obstruction cells or collapse vector distinctions before numerical refinement.

Multi-objective verification and multidimensional energy or mean-payoff objectives study vector criteria in probabilistic, nondeterministic, or game settings \cite{etessami2008,forejt2011,chatterjee2012}. Those settings usually include nondeterminism, probability, adversarial strategy, or infinite-horizon objectives. The main results here are for deterministic finite-word systems; the infinite-word discussion is a scope boundary.

\section{Cones, Lineality, and Intrinsic Observations}

Let $V=\mathbb Q^d$. A rational polyhedral cone $K\subseteq V$ has lineality space
\[
        L_K=K\cap(-K).
\]
The induced preorder identifies vectors up to $L_K$: if $x-y\in L_K$, then $x\le_K y$ and $y\le_K x$. Geometrically, lineality directions are order-neutral. Any residual difference that remains inside $L_K$ under all continuations is invisible to every conic obstruction and is therefore collapsed by the conic quotient. Thus non-pointed cones do not merely add a technical nuisance; they specify directions that the preorder itself declares indistinguishable. The dual cone
\[
        K^\ast=\{\lambda\in V^\ast: \lambda(x)\ge0\text{ for all }x\in K\}
\]
annihilates $L_K$.

Let $\mathcal E(K^\ast)$ denote the set of extreme rays of $K^\ast$. For a rational polyhedral cone this set is finite. Choose one nonzero primitive rational generator $r$ for each extreme ray. The signs of $r(x)$ are independent of the scaling of $r$.

\begin{definition}[Conic obstruction cell]
For $x\in V$, define the intrinsic conic sign profile
\[
        \sigma_K(x)=\bigl(\sign(r(x))\bigr)_{r\in\mathcal E(K^\ast)}
        \in\{-1,0,+1\}^{\mathcal E(K^\ast)}.
\]
The conic obstruction cell of $x$ is the cell of the hyperplane arrangement
\[
        r(x)=0,\qquad r\in\mathcal E(K^\ast),
\]
containing $x$. Equivalently, it is the sign profile $\sigma_K(x)$.
\end{definition}

\begin{definition}[Intrinsic conic covector map]
The intrinsic conic covector map of $K$ is
\[
        \kappa_K:V\to\{-1,0,+1\}^{\mathcal E(K^\ast)},
        \qquad
        \kappa_K(x)=\sigma_K(x).
\]
For a vector-valued language $h$, the induced residual covector stratification is the two-variable map
\[
        \mathcal S_K(h)(u,z)=\kappa_K(\partial_u h(z)).
\]
The conic observation quotient studied below is the right-stable quotient induced by this stratification.
\end{definition}

\begin{definition}[Restricted covector configuration]
Let
\[
        W_h=\operatorname{span}_{\mathbb Q}\{\partial_u h(z):u,z\in\Sigma^\ast\}
\]
be the residual span of $h$. The restricted conic covector configuration is the list of nonzero restrictions
\[
        \mathcal R_K(h)=\{r|_{W_h}:r\in\mathcal E(K^\ast),\ r|_{W_h}\ne0\},
\]
taken up to positive scaling and relabeling. Its sign vectors
\[
        \bigl(\sign(r(x))\bigr)_{r\in\mathcal R_K(h)},\qquad x\in W_h,
\]
form a realizable central sign configuration; equivalently, they are the covectors of a realizable oriented matroid in the standard sense. No oriented-matroid theorem is used below. The point of this language is only that deleting a supplied covector, restricting from $V$ to $W_h$, and contracting lineality have their usual sign-configuration meanings. The conic carrier depends only on this restricted sign configuration together with the residual right action.
\end{definition}

\begin{proposition}[Invariance under restricted covector isomorphism]
Suppose two cones $K_1,K_2$ induce isomorphic restricted covector configurations on $W_h$, with a fixed isomorphism preserving sign vectors up to relabeling of covectors. Then the corresponding conic congruences are equal after applying that relabeling, and the carriers $Q_{K_1}(h)$ and $Q_{K_2}(h)$ are isomorphic, relative to the chosen relabeling, as right-action carriers with cell observations.
\end{proposition}

\begin{proof}
The defining data of $u\equiv_K v$ is the sign vector of every residual future $\partial_{ux}h(z)$ under the restricted covectors. If the restricted sign configurations are isomorphic, equality of all sign vectors for $K_1$ is equivalent to equality of all sign vectors for $K_2$. Hence the equivalence relations coincide up to relabeling of observation coordinates, and the quotient carriers are isomorphic.
\end{proof}

The profile is intrinsic in the presentation-independent sense used throughout: it is unchanged by positive rescaling of extreme-ray generators and by relabeling their order. If $K$ is not pointed, directions in $L_K$ are invisible to all $r\in K^\ast$.

\begin{proposition}[Lineality reduction]
Let $\bar V=V/L_K$, let $\pi:V\to\bar V$ be the quotient map, and let $\bar K=\pi(K)$. Then $\bar K$ is pointed and
\[
        K^\ast \cong \bar K^\ast
\]
under pullback along $\pi$. Consequently,
\[
        \sigma_K(x)=\sigma_{\bar K}(\pi x)
\]
up to relabeling of extreme rays, and for every vector series $h$,
\[
        u\equiv_K v
        \quad\Longleftrightarrow\quad
        u\equiv_{\bar K} v\text{ for the quotient series }\pi\circ h.
\]
Thus all lineality directions are collapsed before any conic observation carrier is formed.
\end{proposition}

\begin{proof}
Every functional in $K^\ast$ vanishes on $L_K$, so it factors uniquely through $\bar V$. Conversely, every nonnegative functional on $\bar K$ pulls back to a nonnegative functional on $K$. This identifies the dual cones and their extreme rays up to scaling. The equality of sign profiles follows immediately. Applying the same equality to every incremental residual $\partial_u h(z)$ gives the congruence statement.
\end{proof}

\begin{corollary}[Residual lineality collapse]
If two prefixes $u,v$ satisfy
\[
        \partial_{ux}h(z)-\partial_{vx}h(z)\in L_K
        \qquad\text{for every }x,z\in\Sigma^\ast,
\]
then $u\equiv_K v$. Equivalently, the conic observation quotient depends only on the residual futures after applying the quotient map $\pi:V\to V/L_K$.
\end{corollary}

\begin{proof}
Every extreme dual ray vanishes on $L_K$, so the two residual futures have the same extreme-ray signs for every right context and continuation. The last statement is the lineality reduction applied to all residual futures.
\end{proof}
We use the phrase extreme-dual-ray arrangement in this precise sense: the hyperplane arrangement $r(x)=0$ induced by the extreme rays of $K^\ast$ is the central arrangement that records which extremal dual faces can be active or negative on a residual vector. The construction uses this cone-induced extreme-ray sign arrangement, not an arbitrary coordinate sign abstraction.

\begin{definition}[Negative and active extreme faces]
For $x\in V$, define
\[
        E_-(x)=\{r\in\mathcal E(K^\ast): r(x)<0\},
        \qquad
        E_0(x)=\{r\in\mathcal E(K^\ast): r(x)=0\}.
\]
The negative extreme cone is $\cone(E_-(x))$, and the active extreme cone is $\cone(E_0(x))$. These cones encode which extremal Farkas directions can witness violation and which extremal directions are tight.
\end{definition}

For a complete finite inequality representation, this reduces to a sign pattern over those inequalities only when the inequalities are exactly the extreme rays of $K^\ast$ up to positive scaling. Redundant inequalities may refine the representation, but not the intrinsic conic cell.

\section{Conic Covector-Congruence}

Let $h:\Sigma^\ast\to V$ be a vector-valued quantitative language. For a prefix $u$, recall
\[
        \partial_u h(z)=h(uz)-h(u).
\]
The following relation is algebraically a genuine right congruence. It is not intended to be weaker than a congruence; the adjective covector indicates that the preserved observations are conic covector cells rather than full vector-valued residuals.

\begin{definition}[Residual future covector languages]
For a prefix $u$ and a conic sign profile $\pi\in\{-1,0,+1\}^{\mathcal E(K^\ast)}$, define the residual future facial language
\[
        F_{u,\pi}=\{z\in\Sigma^\ast: \sigma_K(\partial_u h(z))=\pi\}.
\]
For a right context $x\in\Sigma^\ast$, $F_{ux,\pi}$ is the same language after the prefix has first been moved to $ux$.
\end{definition}

\begin{definition}[Conic covector congruence]
For $u,v\in\Sigma^\ast$, define
\[
        u\equiv_K v
        \quad\Longleftrightarrow\quad
        F_{ux,\pi}=F_{vx,\pi}
        \quad\text{for every }x\in\Sigma^\ast\text{ and every profile }\pi .
\]
Equivalently,
\[
        \sigma_K(\partial_{ux} h(z))=
        \sigma_K(\partial_{vx} h(z))
        \quad\text{for every }x,z\in\Sigma^\ast .
\]
The quotient $Q_K(h)=\Sigma^\ast/\equiv_K$ is the conic residual quotient of $h$.
\end{definition}

\begin{lemma}[Right congruence]
The relation $\equiv_K$ is a right congruence: if $u\equiv_K v$, then $ua\equiv_K va$ for every $a\in\Sigma$.
\end{lemma}

\begin{proof}
Assume $u\equiv_K v$. For $ua$ and $va$, the defining condition requires equality of $F_{uax,\pi}$ and $F_{vax,\pi}$ for every $x$ and $\pi$. This is exactly the defining condition for $u\equiv_K v$ applied to the right context $ax$. Hence $ua\equiv_K va$.
\end{proof}

For a profile $\pi$, one may also consider the global language
\[
        L_\pi=\{w\in\Sigma^\ast:\sigma_K(h(w))=\pi\}.
\]
The carrier below is not built from $L_\pi$ alone. It is built from the explicit residual languages
\[
        F_{u,\pi}=\{z:\sigma_K(h(uz)-h(u))=\pi\},
\]
indexed by the prefix $u$. This notation is used throughout; no inverse-language shorthand is needed.

The following proposition records the carrier of the cone-induced observation family. The proof is the standard Myhill--Nerode right-congruence argument for a chosen observation family. The nonstandard part is not this minimization proof pattern; it is the source and calculus of the observations: the family is generated by the extreme-dual covectors of $K$, restricted to residual futures, and later compared by cone-refinement operations.

\begin{proposition}[Carrier of the conic residual stratification]
The quotient $Q_K(h)$ is the coarsest right congruence preserving the cone-induced residual stratification. More precisely, for any right congruence $\rho$ on $\Sigma^\ast$, the following are equivalent:
\begin{enumerate}[label=(\roman*)]
\item $\rho$ refines $\equiv_K$;
\item whenever $u\rho v$, one has
\[
        F_{ux,\pi}=F_{vx,\pi}
        \quad\text{for every right context }x\in\Sigma^\ast
        \text{ and every conic profile }\pi .
\]
\end{enumerate}
Equivalently, every deterministic carrier whose state congruence preserves all languages $F_{ux,\pi}$ factors through $Q_K(h)$. Any proper coarsening of $Q_K(h)$ merges two prefixes separated by a concrete residual future facial language.
\end{proposition}

\begin{proof}
By definition, $u\equiv_K v$ holds exactly when the equality in (ii) holds for all right contexts and profiles. Hence any right congruence satisfying (ii) identifies only pairs already identified by $\equiv_K$, so it refines $\equiv_K$. Conversely, if $\rho$ refines $\equiv_K$ and $u\rho v$, then $u\equiv_K v$, and the same equalities of residual future covector languages follow.

For the factorization claim, let $C$ be a deterministic carrier and let $\rho_C$ be the right congruence induced by reaching the same state of $C$. If $C$ preserves all languages $F_{ux,\pi}$, then $\rho_C$ satisfies (ii), hence $\rho_C$ refines $\equiv_K$. Therefore the map from reachable states of $C$ to $Q_K(h)$ is well defined. If a proper coarsening of $Q_K(h)$ merges $[u]$ and $[v]$, then $u\not\equiv_K v$; by definition there are $x,z$ and a profile $\pi$ such that $z$ belongs to exactly one of $F_{ux,\pi}$ and $F_{vx,\pi}$. This is a concrete residual future facial language separating the merged states.
\end{proof}

\begin{definition}[Conic observation carrier]
Assume $Q_K(h)$ is finite. The conic observation carrier of $h$ over $K$ is
\begin{equation}
        \mathcal C_K(h)=\bigl(Q_K(h),(T_a)_{a\in\Sigma},(\omega_z)_{z\in\Sigma^\ast}\bigr),
\end{equation}
where
\[
        T_a([u])=[ua]
        \qquad\text{and}\qquad
        \omega_z([u])=\sigma_K(\partial_u h(z)).
\]
The maps $T_a$ are the right action of the free monoid, while $\omega_z$ records the conic covector cell of a residual future. A morphism of carriers preserves the transition maps and the cell observations.
\end{definition}

The transition monoid generated by the $T_a$ is the ordinary syntactic transition monoid of the residual covector language family. The additional structure is the family of cell-valued observations $\omega_z$. Thus the term carrier is used only for the finite right-action representation of the conic observation quotient.

\section{Row View and Stratification Refinement}

The previous carrier proposition identifies a right congruence. The following reformulation records the same object as a row kernel: entries are finite conic cells rather than values in a field or semiring.

Let
\[
        \Pi_K=\{\sigma_K(x):x\in V\}\subseteq\{-1,0,+1\}^{\mathcal E(K^\ast)}
\]
be the set of realizable conic sign cells. Define the row map
\[
        \mathsf H_K(h):\Sigma^\ast\longrightarrow \Pi_K^{\Sigma^\ast\times\Sigma^\ast},
        \qquad
        \mathsf H_K(h)(u)(x,z)=\sigma_K(\partial_{ux}h(z)).
\]

\begin{proposition}[Row kernel]
For all words $u,v$,
\[
        u\equiv_K v
        \quad\Longleftrightarrow\quad
        \mathsf H_K(h)(u)=\mathsf H_K(h)(v).
\]
Hence $Q_K(h)$ is the row set of this cell-valued map. In particular, $Q_K(h)$ has finite index if and only if $\mathsf H_K(h)$ has finite image.
\end{proposition}

\begin{proof}
The equality of rows says exactly that for every right context $x$ and continuation $z$, the profiles $\sigma_K(\partial_{ux}h(z))$ and $\sigma_K(\partial_{vx}h(z))$ agree. This is the defining condition of $u\equiv_K v$. The finite-index statement follows because quotient classes are precisely distinct rows.
\end{proof}

The construction is monotone in the amount of dual geometry retained. If $S$ is any finite set of nonzero dual functionals, define $\sigma_S(x)=(\sign(s(x)))_{s\in S}$ and the corresponding congruence $\equiv_S$ by replacing $\sigma_K$ with $\sigma_S$ in the definition above.

\begin{proposition}[Stratification refinement law]
For finite dual test sets $S,T$,
\[
        \equiv_{S\cup T}=\equiv_S\cap\equiv_T.
\]
Consequently, if $S\subseteq T$, then
\[
        \equiv_T\subseteq\equiv_S,
\]
and there is a canonical surjection $Q_T(h)\twoheadrightarrow Q_S(h)$ whenever both quotients are finite.
\end{proposition}

\begin{proof}
Two prefixes have the same $(S\cup T)$-profiles for all right contexts and continuations exactly when they have the same $S$-profiles and the same $T$-profiles for all right contexts and continuations. This proves the identity of congruences. The monotonicity and quotient map follow immediately.
\end{proof}

Taking $S=
\mathcal E(K^\ast)$ gives the intrinsic conic congruence. Taking a proper subset of extreme rays gives a coarser, representation-limited abstraction. Adding redundant inequalities can refine a chosen sign presentation, but it does not alter the presentation-independent quotient unless it introduces hyperplanes not identified with the extreme-dual-ray arrangement. Thus the algebraic content is not the existence of a Myhill--Nerode quotient; it is the stratification calculus that sends dual geometric data to right congruences.

\subsection{Observable covectors and refinement}

The previous refinement law can be organized as an observation-refinement correspondence. For any set $S$ of nonzero covectors, let $\equiv_S$ denote equality of all residual signs induced by covectors in $S$; finite $S$ gives a finite observation alphabet, but the definition itself does not require finiteness. For a right congruence $\rho$, define the set of covectors observable through $\rho$ by
\[
        \mathrm{Obs}_h(\rho)=
        \left\{r\in V^\ast\setminus\{0\}:
        u\rho v\Rightarrow
        \sign r(\partial_{ux}h(z))=
        \sign r(\partial_{vx}h(z))
        \text{ for all }x,z\in\Sigma^\ast
        \right\}.
\]
This set is closed under positive rescaling of covectors. It records precisely which Farkas directions can be read without splitting the states of $\rho$.

\begin{theorem}[Observable-covector correspondence]
Let $S\subset V^\ast\setminus\{0\}$ be any set of nonzero covectors and let $\rho$ be a right congruence on $\Sigma^\ast$. Then
\[
        S\subseteq \mathrm{Obs}_h(\rho)
        \quad\Longleftrightarrow\quad
        \rho\subseteq \equiv_S .
\]
Consequently, the assignment $S\mapsto\equiv_S$ is antitone: adding covectors refines the congruence. The operator
\[
        \operatorname{cl}_h(S)=\mathrm{Obs}_h(\equiv_S)
\]
will be used below.
\end{theorem}

\begin{proof}
If $S\subseteq\mathrm{Obs}_h(\rho)$ and $u\rho v$, then every $s\in S$ has the same residual sign on $\partial_{ux}h(z)$ and $\partial_{vx}h(z)$ for all $x,z$. Hence $u\equiv_S v$, so $\rho\subseteq\equiv_S$.
Conversely, if $\rho\subseteq\equiv_S$ and $s\in S$, then $u\rho v$ implies $u\equiv_S v$, and the defining equality of all $S$-sign residuals gives $s\in\mathrm{Obs}_h(\rho)$.
The antitonicity is the case $S\subseteq T$: then $S\subseteq\mathrm{Obs}_h(\equiv_T)$, hence $\equiv_T\subseteq\equiv_S$.
\end{proof}

\begin{lemma}[Closure preserves the induced congruence]
For every set $S$ of nonzero covectors,
\[
        \equiv_{\operatorname{cl}_h(S)}=\equiv_S .
\]
\end{lemma}

\begin{proof}
Let $C=\operatorname{cl}_h(S)=\mathrm{Obs}_h(\equiv_S)$. Since $S\subseteq C$, antitonicity gives $\equiv_C\subseteq\equiv_S$.
For the reverse inclusion, $C\subseteq\mathrm{Obs}_h(\equiv_S)$ holds by definition of $C$. Applying the observable-covector correspondence with covector set $C$ and right congruence $\equiv_S$ gives $\equiv_S\subseteq\equiv_C$. Hence the two congruences are equal.
\end{proof}

\begin{proposition}[Observable-covector closure]
The operator $\operatorname{cl}_h(S)=\mathrm{Obs}_h(\equiv_S)$ is a closure operator on covector sets up to positive scaling: it contains $S$, is monotone in $S$, and satisfies
\[
        \operatorname{cl}_h(\operatorname{cl}_h(S))=\operatorname{cl}_h(S).
\]
\end{proposition}

\begin{proof}
Extensivity follows from the observable-covector correspondence applied to the tautology $\equiv_S\subseteq\equiv_S$, giving $S\subseteq\mathrm{Obs}_h(\equiv_S)$.
For monotonicity, suppose $S\subseteq T$. Antitonicity gives $\equiv_T\subseteq\equiv_S$. If $r\in\operatorname{cl}_h(S)$, then $r$ is observable through $\equiv_S$; since $\equiv_T$ is finer, $r$ is also observable through $\equiv_T$. Thus $\operatorname{cl}_h(S)\subseteq\operatorname{cl}_h(T)$.
For idempotence, the previous lemma gives $\equiv_{\operatorname{cl}_h(S)}=\equiv_S$, so
\[
        \operatorname{cl}_h(\operatorname{cl}_h(S))
        =\mathrm{Obs}_h(\equiv_{\operatorname{cl}_h(S)})
        =\mathrm{Obs}_h(\equiv_S)
        =\operatorname{cl}_h(S).
\]
\end{proof}

The theorem identifies, for a given residual quotient, the largest dual covector geometry that the quotient can observe without further refinement. This gives the advertised coupling between cone geometry and residual algebra.

\subsection{Cone-refinement calculus}

Different cones can induce the same or different covector partitions. To avoid relying on cone inclusion alone, define an observation refinement preorder on cones by
\[
        K_1\preceq_{\mathrm{obs}}K_2
        \quad\Longleftrightarrow\quad
        \sigma_{K_2}(x)=\sigma_{K_2}(y)
        \Rightarrow
        \sigma_{K_1}(x)=\sigma_{K_1}(y)
        \quad\text{for all }x,y\in V .
\]
Thus $K_2$ has at least as fine an extreme-dual sign arrangement as $K_1$.

\begin{theorem}[Cone-refinement morphisms]
If $K_1\preceq_{\mathrm{obs}}K_2$, then
\[
        \equiv_{K_2}\subseteq\equiv_{K_1}.
\]
Hence, whenever both quotients are finite, there is a canonical surjective morphism of right-action carriers
\[
        Q_{K_2}(h)\twoheadrightarrow Q_{K_1}(h),
        \qquad [u]_{K_2}\mapsto [u]_{K_1}.
\]
If $K_1$ and $K_2$ induce the same extreme-dual sign arrangement, then their conic observation carriers are isomorphic for every $h$ by the identity on prefix classes after the chosen relabeling of observation coordinates.
\end{theorem}

\begin{proof}
Assume $u\equiv_{K_2}v$. Then all $K_2$-profiles of $\partial_{ux}h(z)$ and $\partial_{vx}h(z)$ agree for all $x,z$. By the definition of $K_1\preceq_{\mathrm{obs}}K_2$, equality of $K_2$-profiles implies equality of $K_1$-profiles pointwise. Thus $u\equiv_{K_1}v$. The quotient morphism is well defined by this inclusion of congruences. If the arrangements refine each other, both inclusions hold and the quotients are isomorphic.
\end{proof}

This calculus separates presentation-independent conic structure from representation artifacts. Positive rescaling and relabeling of extreme-ray generators do not change the carrier. Adding non-extreme redundant inequalities may refine a chosen presentation, but it corresponds to changing the observation set, not to changing the presentation-independent cone-induced carrier.

Observation refinement is not the same relation as cone inclusion. If $K_1\subseteq K_2$, then $K_2^\ast\subseteq K_1^\ast$, so the order is stronger but the extreme-ray arrangements need not refine each other as sign partitions of $V$. For example, in $\mathbb Q^2$ the positive orthant uses coordinate tests $(1,0)$ and $(0,1)$, while the wedge $K_\triangle$ uses the mixed tests $(-1,1)$ and $(1,1)$. Neither sign arrangement globally refines the other: coordinate signs can distinguish vectors that the wedge tests place in the same cell, and conversely the wedge tests can distinguish vectors with the same coordinate signs. The preorder $\preceq_{\mathrm{obs}}$ therefore records arrangement refinement, not set inclusion of cones.

\begin{lemma}[Finite rational refinement test]
Let $S_1,S_2\subset(\mathbb Q^d)^\ast$ be finite covector families, and write $\sigma_i$ for their sign maps. Then $S_2$ refines $S_1$,
\[
        \sigma_2(x)=\sigma_2(y)\Rightarrow\sigma_1(x)=\sigma_1(y)
        \quad\text{for all }x,y\in\mathbb Q^d,
\]
if and only if every realizable sign cell of $S_2$ is contained in a sign cell of $S_1$. This condition is decidable over the rationals by enumerating sign vectors $\eta\in\{-1,0,+1\}^{S_2}$ and, for each $s\in S_1$, checking whether the rational linear system defining the $\eta$-cell admits points with two different signs of $s$.
\end{lemma}

\begin{proof}
The first equivalence is just the definition of refinement in cell language: equality of $S_2$-profiles means membership in the same $S_2$-cell, and equality of $S_1$-profiles means containment in one $S_1$-cell.
For decidability, an $S_2$ sign vector $\eta$ defines a rational polyhedral cone by constraints $t(x)>0$, $t(x)=0$, or $t(x)<0$ for $t\in S_2$. For a fixed $s\in S_1$, failure of constancy is witnessed by two points $x,y$ in the same $\eta$-cell and two different signs for $s(x)$ and $s(y)$. Thus one checks the finitely many pairs of requested signs in $\{-1,0,+1\}^2$ with distinct entries. Each check is a rational linear feasibility problem with equalities, non-strict inequalities, and strict inequalities. Strict rational linear inequalities are decidable by standard rational linear programming, equivalently by adding a rational slack variable and normalizing homogeneous strict constraints. Since there are finitely many $\eta$, $s$, and sign pairs, the test terminates.
\end{proof}

A direct decision example may be useful. Let $S_1=\{(1,0)\}$, $S_2=\{(1,0),(0,1)\}$, $S_3=\{(2,0)\}$, and $S_4=\{(1,1)\}$ as sign-test families on $\mathbb Q^2$. The arrangements for $S_1$ and $S_3$ are the same, since positive rescaling does not change signs. The arrangement for $S_2$ is finer than that for $S_1$, since equality of both coordinate signs implies equality of the first coordinate sign. The arrangements for $S_1$ and $S_4$ are incomparable: $(1,-2)$ and $(1,2)$ have the same $S_1$ sign but different $S_4$ signs, while $(1,-2)$ and $(-1,2)$ have the same $S_4$ sign but different $S_1$ signs. This is the practical test for $\preceq_{\mathrm{obs}}$: one must check whether every cell of the proposed finer arrangement is contained in a cell of the proposed coarser arrangement.

\begin{proposition}[Cell versus value collapse]
Let $\equiv_{\mathrm{vec}}$ denote equality of all vector residual futures. Then
\[
        \equiv_{\mathrm{vec}}\subseteq \equiv_K.
\]
The inclusion is strict if and only if there exist prefixes $u,v$ such that some vector residual future differs,
\[
        \partial_{ux}h(z)\ne \partial_{vx}h(z),
\]
for some $x,z$, while all such futures remain inside the same conic sign cells:
\[
        \sigma_K(\partial_{ux}h(z))=
        \sigma_K(\partial_{vx}h(z))
        \qquad\text{for every }x,z.
\]
\end{proposition}

\begin{proof}
Vector residual equality implies equality after applying every extreme dual ray and hence equality of conic sign profiles. Strictness is precisely the existence of a pair distinguished by vector values but not by the cell-valued rows.
\end{proof}

\section{Relation to Numerical Certificates}

The conic quotient preserves obstruction cells, not magnitudes. Difference-constraint certificates depend on accumulated scalar values. Hence the conic quotient cannot replace the numerical carrier for Bellman--Ford feasibility unless additional data is retained.

Let $\mathcal E(K^\ast)=\{r_1,
\ldots,r_m\}$. The exact numerical projected residual congruence is
\[
        u\equiv_{\mathrm{num},K} v
        \quad\Longleftrightarrow\quad
        r_i(\partial_u h(z))=r_i(\partial_v h(z))
        \quad\forall i,z.
\]
This is the ordinary residual congruence of the projected series
\[
        (r_1h,
        \ldots,r_mh).
\]

\begin{proposition}[Three layers]
The following implications hold:
\[
        u\equiv_{\text{vector}}v
        \Longrightarrow
        u\equiv_{\mathrm{num},K}v
        \Longrightarrow
        u\equiv_K v.
\]
Both implications can be strict.
\end{proposition}

\begin{proof}
Vector residual equality implies equality after applying every dual extreme ray. Numerical equality implies equality of signs. Strictness occurs when vector residuals differ in a direction annihilated by $K^\ast$, or when scalar values differ without crossing any hyperplane $r_i(x)=0$.
\end{proof}

\begin{example}[Same covector type, different numerical feasibility data]
Let $K=\mathbb Q_{\ge0}$, so $K^\ast=\mathbb Q_{\ge0}$ has one extreme ray. Scalar weights $1$ and $2$ have the same conic cell: both are positive. A covector carrier identifies them. But a path of total cost $1$ and a path of total cost $2$ give different shortest-path potentials when combined with terminal offsets. Thus conic type does not determine potential values or numerical feasibility data. It only determines absence or presence of sign-level obstruction.
\end{example}

\section{Verification-Style Example: Resource Reserve Monitor}

This example shows the complete qualitative pipeline on a small deterministic resource monitor. The system has states $p_0,p_1,p_2$ and alphabet $\Sigma=\{a,b,c\}$. The vector value records a two-dimensional gap: the first coordinate is accumulated risk and the second coordinate is reserve. We use the wedge cone $K_\triangle$ with rays
\[
        r^- =(-1,1),\qquad r^+=(1,1).
\]
Transitions and vector increments are
\[
\begin{array}{c|ccc}
 & a & b & c\\
\hline
p_0 & p_1:(1,2) & p_2:(4,2) & p_0:(2,2)\\
 p_1 & p_1:(2,4) & p_2:(4,2) & p_0:(-2,2)\\
 p_2 & p_2:(1,1) & p_2:(4,2) & p_0:(1,3).
\end{array}
\]
The conic profiles of the displayed increments are obtained by applying $(r^-,r^+)$. Thus $(1,2)$ and $(2,4)$ are both $(+,+)$, $(4,2)$ is $(-,+)$, $(2,2)$ is $(0,+)$, $(-2,2)$ is $(+,0)$, and $(1,3)$ is $(+,+)$.

A bounded horizon-$1$ conic pass already finds that any $b$-transition is a qualitative obstruction, because its profile is $(-,+)$. This is a concrete Farkas witness: $r^-(4,2)=-2<0$, so that increment lies outside $K_\triangle$. The same pass also merges $a$-loops with increments $(1,2)$ and $(2,4)$ at the conic level, even though a numerical residual carrier would keep them separate. For horizon $2$, accumulated futures refine the result: for instance, from $p_0$ the word $ac$ accumulates $(1,2)+(-2,2)=(-1,4)$ with profile $(+,+)$, while from $p_1$ the word $ac$ accumulates $(2,4)+(-2,2)=(0,6)$ with profile $(+,+)$; these remain conically indistinguishable for this word. In contrast, contexts containing $b$ expose the negative cell immediately.

The exact numerical verification fallback is then run only on the scalarized problems for unresolved nonnegative cells. For the $r^-$ scalarization, the $b$ edge has scalar cost $-2$, so a one-edge counterexample is produced before a full potential computation. For cells that remain nonnegative under both rays, Bellman--Ford or another scalar potential method is still required if terminal offsets or magnitude-sensitive constraints are present. This example illustrates the role of the conic observation quotient: it exposes cone-geometric obstruction cells and can reduce the set of behaviours passed to numerical analysis.

\section{Qualitative Screening Role}

The conic observation carrier can be used as an abstract domain for the cone order. Its elements are conic cells rather than numerical path values. The abstraction map sends a vector value $x$ to its intrinsic profile $\sigma_K(x)$. If some extreme ray $r\in\mathcal E(K^\ast)$ satisfies $r(x)<0$, then $x\notin K$; such a negative component is a concrete Farkas witness. Therefore a reachable word whose accumulated value lies in a negative cell is a sound qualitative counterexample to cone inclusion.

This screening role is one-sided in the algorithmic sense. A negative cell gives a real obstruction. Absence of a negative cell in the explored conic quotient leaves magnitude-sensitive obligations unresolved, because scalar potentials depend on exact values and terminal offsets. A typical workflow is therefore to use the conic observation quotient as a qualitative pre-analysis: construct a bounded or exact quotient, search for negative cells, and then pass the unresolved part to the scalar potential procedure.

The useful regimes are correspondingly concrete. If accumulated futures cross an extreme-dual hyperplane at small horizons, the screen produces short Farkas witnesses. If futures stay in the same open positive cell, it may still collapse sign-equivalent behaviours but does not remove the scalar fallback. If many futures lie near faces of the cone, increasing the horizon can change both the partition and the number of witnesses. Section~\ref{sec:computational-evaluation} reports these cases with deterministic ancillary models.

\section{Collapse and Strict Separation}

The conic quotient coincides with the vector residual quotient only when conic cells separate every residual difference.

\begin{proposition}[Collapse criterion]
The conic quotient $Q_K(h)$ equals the ordinary vector residual quotient if and only if for every pair of distinct vector residuals there are words $x,z$ and an extreme ray $r\in\mathcal E(K^\ast)$ such that
\[
        \sign(r(\partial_{ux} h(z)))
        \ne
        \sign(r(\partial_{vx} h(z))).
\]
Equivalently, no two distinct vector residuals remain in the same extreme-ray sign cells in all right contexts and continuations.
\end{proposition}

\begin{proof}
If the stated separation property holds, any two distinct vector residuals are distinguished by the conic covector congruence, so the quotients coincide. If it fails, two distinct vector residuals have identical conic profiles in all right contexts and continuations and are identified by $Q_K(h)$, making the conic quotient strictly coarser.
\end{proof}

\begin{example}[Non-orthant strict coarsening]
Let
\[
        K_\triangle=\{(x_1,x_2):x_2-x_1\ge0,
        \ x_2+x_1\ge0\}.
\]
The extreme dual rays can be represented by $r^- =(-1,1)$ and $r^+=(1,1)$. The vectors $(1,2)$ and $(2,4)$ have different numerical values but identical signs under $r^-$ and $r^+$: both are strictly positive for both tests. A vector residual quotient can distinguish them. The conic covector quotient cannot, unless some continuation moves one vector across an extreme-dual-ray hyperplane.
\end{example}

This is an obstruction-type example: within the same open conic cell, facial geometry does not see magnitude.

\section{Construction Regimes and Complexity}

The procedures in this section compute sign-cell abstractions of accumulated residual futures. A negative sign under an extreme dual ray is a sound Farkas obstruction, while nonnegative conic cells leave scalar potential feasibility, shortest-path bounds, and terminal-offset constraints to the numerical layer.

The definition of $\equiv_K$ concerns signs of accumulated future residuals. Edge-local signs are insufficient in general: $+10+(-9)$ and $+1+(-2)$ have the same one-step sign pattern but different accumulated signs. Thus a correct construction must compare accumulated future profiles.

\subsection{A computable bounded-horizon fragment}

For a fixed horizon $H$, define
\[
        u\equiv_{K,H}v
        \quad\Longleftrightarrow\quad
        \sigma_K(\partial_{ux}h(z))=
        \sigma_K(\partial_{vx}h(z))
\]
for every $x,z$ with $|xz|\le H$. This finite approximation is useful in bounded model checking and gives a completely explicit algorithm.

For the finite weighted carrier notation in Procedure~1, let
\[
        h_p(a_1\cdots a_k)
        =
        \sum_{i=1}^k g(q_{i-1},a_i)+\tau(q_k),
        \qquad
        q_0=p,\quad q_i=\delta(q_{i-1},a_i).
\]
Thus $h_p(\epsilon)=\tau(p)$. If $p_x=\delta(p,x)$, then the accumulated residual used in the bounded signatures is
\[
        \partial_{px}h(z)=h_{p_x}(z)-h_{p_x}(\epsilon).
\]
Equivalently, for $z=a_1\cdots a_k$ from $p_x$, it is the sum of the edge increments along $z$, plus $\tau(q_k)-\tau(p_x)$.

\begin{procedure}[Bounded-horizon conic quotient]
Input: a finite reachable deterministic weighted carrier $P_R=(P,\Sigma,\delta,g,\tau)$ over $\mathbb Q^d$, where $g(p,a)$ is the vector increment on edge $(p,a)$ and $\tau$ is the terminal offset used to evaluate finite futures; a complete list of primitive extreme dual rays $\mathcal E(K^\ast)=\{r_1,\ldots,r_e\}$; and a horizon $H\ge0$.

Output: the partition $Q_{K,H}(h)$ preserving all conic residual profiles $\sigma_K(\partial_{px}h(z))$ with $|xz|\le H$, the induced quotient transitions when the partition is right-stable, and the list of reachable negative-cell witnesses $(p,x,z,r_i)$ found inside the horizon.

\begin{enumerate}[label=\arabic*.]
\item Enumerate all pairs $(x,z)$ with $|xz|\le H$.
\item For each reachable state $p$ and each pair $(x,z)$, compute $p_x=\delta(p,x)$ and then compute $\partial_{px}h(z)=h_{p_x}(z)-h_{p_x}(\epsilon)$ by summing the increments along $z$ and adding $\tau(q_k)-\tau(p_x)$.
\item Form the finite signature
\[
        \Theta_H(p)=
        \bigl(\sigma_K(\partial_{px}h(z))\bigr)_{|xz|\le H}.
\]
\item Identify states with the same signature.
\item Record every coordinate $r_i(\partial_{px}h(z))<0$ as a bounded Farkas witness.
\item If a quotient transition is needed, close the partition under letter successors and recompute signatures until stable.
\end{enumerate}
\end{procedure}

This bounded construction is exact for the horizon-limited relation by enumeration of accumulated futures. If $N=|P|$, $s=|\Sigma|$, $e=|\mathcal E(K^\ast)|$, and vectors have dimension $d$, then the number of word pairs is
\[
        B_H=\sum_{n=0}^H (n+1)s^n,
\]
because a word of total length $n$ can be split as $xz$ in $n+1$ ways. The signature table has $NB_H$ accumulated values and $NeB_H$ ray evaluations. With a dynamic program over the carrier, accumulated values can be computed in $O(NdB_H)$ rational arithmetic operations after the transition table is available. Each ray evaluation is a rational dot product costing $O(d)$ arithmetic operations, so the sign table costs $O(NedB_H)$ arithmetic operations. The final stable partition refinement adds at most the usual deterministic refinement cost, for example $O(sN^2)$ for a simple repeated-splitting implementation. Thus the exponential dependence is exactly in $B_H$, not hidden in an oracle.

\begin{example}[A bounded-horizon run]
Return to the two-state wedge carrier from the introduction. Use the edge labels as additive residual increments and take zero terminal offsets. For $H=1$, the relevant nonempty continuations are $a$ and $b$. Both states see profile $(+,+)$ on $a$, and both see profile $(-,+)$ on $b$. Hence $Q_{K,1}$ has one block.

For $H=2$, accumulated futures separate the states. From $p$, the word $ab$ yields
\[
        (1,2)+(4,2)=(5,4),
        \qquad
        \sigma_{K_\triangle}(5,4)=(-,+).
\]
From $q$, the same word yields
\[
        (2,4)+(4,2)=(6,6),
        \qquad
        \sigma_{K_\triangle}(6,6)=(0,+).
\]
Thus $Q_{K,2}$ splits the two states. This is a small complete run of the bounded algorithm and illustrates why accumulated residuals, rather than edge-local signs, are the data being minimized.
\end{example}

\subsection{Computational evaluation}
\label{sec:computational-evaluation}

The ancillary Rust implementation evaluates Procedure~1 with exact rational arithmetic on deterministic model families. The scope of the runs is the bounded construction; the exponential $B_H$ law above is the relevant scaling law for increasing horizons. The implementation uses normalized integer rationals, and the table was generated by a single \texttt{cargo run --quiet -- eval} run in Rust's default development profile. Runtime is wall-clock time for the conic pass in that bundled prototype.

The families are generated as follows. The resource monitor has $12$ states, alphabet size $4$, dimension $2$, and the wedge rays $(-1,1),(1,1)$. The random grid has $16$ states, alphabet size $4$, the same wedge rays, and deterministic seed $1729$. The positive-cell family has $10$ states and alphabet size $3$, with all increments staying in the open positive wedge cell. The near-boundary family has $14$ states and alphabet size $4$, with tight and near-tight increments around the wedge face. The large-alphabet family has $8$ states and alphabet size $8$. The high-dimensional family has $9$ states, alphabet size $3$, dimension $4$, and six supplied dual rays. Table~\ref{tab:screening-eval} reports the bounded conic screen.

Memory is a machine-independent working-set estimate used by the prototype:
\[
        M_{\mathrm{est}}
        =
        \left\lceil
        \frac{32NB_Hd+NB_He+64W+Ns(8+32d)}{1024}
        \right\rceil\text{ KiB},
\]
where $W$ is the number of stored negative witnesses. The full CSV and JSON outputs also include scalar Bellman--Ford edge scans, relaxations, first witness length, witness density, unique witnessing states, and unresolved conic cells.

\begin{table}[t]
\centering
\small
\begin{tabular}{@{}lrrrrrrr@{}}
\toprule
Family & Horizon & Word pairs & Ray evals & Blocks & Neg. wit. & Runtime & Memory \\
\midrule
resource-monitor & 1 & 9 & 216 & 12 & 10 & 0.68 ms & 11 KiB \\
resource-monitor & 2 & 57 & 1368 & 12 & 90 & 3.42 ms & 54 KiB \\
resource-monitor & 3 & 313 & 7512 & 12 & 540 & 16.78 ms & 280 KiB \\
random-grid & 3 & 313 & 10016 & 16 & 671 & 17.12 ms & 370 KiB \\
positive-cell & 3 & 142 & 2840 & 1 & 0 & 4.33 ms & 94 KiB \\
near-boundary & 1 & 9 & 252 & 6 & 9 & 0.31 ms & 13 KiB \\
near-boundary & 2 & 57 & 1596 & 14 & 87 & 2.12 ms & 61 KiB \\
near-boundary & 3 & 313 & 8764 & 14 & 565 & 11.99 ms & 322 KiB \\
large-alphabet & 2 & 209 & 3344 & 8 & 97 & 3.94 ms & 119 KiB \\
high-dimensional & 3 & 142 & 7668 & 9 & 627 & 10.72 ms & 211 KiB \\
\bottomrule
\end{tabular}
\caption{Bounded-horizon conic screening on deterministic ancillary model families. Word pairs equal $\sum_{n=0}^H(n+1)|\Sigma|^n$, and ray evaluations equal states times word pairs times extreme rays.}
\label{tab:screening-eval}
\end{table}

The table shows three distinct regimes. In the resource monitor, negative witnesses appear already at horizon $1$ and increase with the enumerated futures; the number of blocks does not collapse, so the gain is early obstruction discovery rather than quotient compression. In the near-boundary family, the partition changes from $6$ blocks at horizon $1$ to all $14$ blocks at horizon $2$, showing that accumulated futures can refine the sign abstraction after one-step evidence has been exhausted. In the positive-cell family, all explored futures stay nonnegative and the quotient collapses to one block; here the conic pass is useful as a compression and classification step, but it leaves the scalar fallback obligation intact.

The witness counts should be read together with density and first-witness length, not as raw evidence of scale. In the generated CSV, the resource-monitor horizon-$3$ run has witness density $71$ per mille and first witness length $1$; the positive-cell run has density $0$; and the high-dimensional run has six rays, so the ray-evaluation count grows even with fewer word pairs. The scalar fallback comparison in the JSON output records Bellman--Ford edge scans and relaxations for the same models. For example, the positive-cell family has no negative conic witness and zero scalar relaxations, while the random-grid family has negative witnesses across all states and two scalar rays with negative-cycle evidence. This comparison is a reproducibility check for the screening/fallback split: conic screening helps when it finds early witnesses or collapses cells, and it contributes little beyond compression in all-positive no-obstruction cases.

\paragraph{Mechanized ancillary core.}
The ancillary material has two roles. The Rust code implements the exact-rational bounded construction, deterministic model families, scalar fallback counters, and the evaluation files used in Table~\ref{tab:screening-eval}. The Lean code formalizes the sign layer used by the refinement arguments: \texttt{SameCell} is equality of sign profiles, \texttt{Refines} is arrangement refinement, \texttt{union\_refines\_left}, \texttt{union\_refines\_right}, and \texttt{sameCell\_union\_iff} formalize refinement by adding observations, and \texttt{neg\_not\_nonnegative} captures the negative-cell soundness step. The compiled \texttt{.olean} and \texttt{.ilean} artifacts are included as curated build evidence.

\subsection{Unbounded construction with a separation procedure}

For the unbounded relation, exact construction is characterization-level unless a separation procedure is available for the chosen class of weighted systems.

\begin{definition}[Future-separation procedure]
For two reachable states $p,q$ of a deterministic carrier for $h$, a conic future-separation procedure is sound if every returned witness is correct, and complete if \textsc{none} is returned only when no witness exists. Its output is either a witness $(r,z)$ with $r\in\mathcal E(K^\ast)$ and $z\in\Sigma^\ast$ such that
\[
        \sign(r(\partial_p h(z)))\ne \sign(r(\partial_q h(z))),
\]
or the answer \textsc{none}, certifying that no such witness exists.
\end{definition}

\begin{procedure}[Exact conic quotient with future separation]
Input: a reachable deterministic weighted carrier $P_R$, alphabet $\Sigma$, extreme rays $\mathcal E(K^\ast)$, and a conic future-separation procedure whose soundness and completeness are established for the chosen model class.

Output: the exact conic quotient $Q_K(h)$ on the reachable carrier, together with the separating witnesses used for the performed splits.

\begin{enumerate}[label=\arabic*.]
\item Initialize all reachable states in one block.
\item Within each block, compare state pairs using the separation procedure; split a block whenever a witness is returned.
\item Enforce right stability: split any block containing states whose $a$-successors lie in different blocks for some $a\in\Sigma$.
\item Iterate until no split occurs.
\end{enumerate}
\end{procedure}

\begin{proposition}[Soundness of exact construction]
If the future-separation procedure is sound and complete, the stable partition returned by Procedure~2 is exactly $\equiv_K$ on the reachable carrier.
\end{proposition}

\begin{proof}
A split caused by the procedure is justified by a concrete continuation and extreme ray, so equivalent states are never separated. Right-stability splits are necessary because $\equiv_K$ is a right congruence. At a fixed point, no two states in the same block admit a separating future witness, and successor blocks agree. Hence all residual future covector languages agree inside each block. Conversely, every separation performed by the procedure corresponds either to a direct future witness or to a right-context propagation of such a witness.
\end{proof}

\subsection{Cost components}

There are four distinct costs. First, if the extreme rays are not supplied, converting a cone representation to $\mathcal E(K^\ast)$ has preprocessing cost $P_K(d,L)$ and may have output size exponential in $d$ \cite{avisfukuda1992,fukudaprodon1996,ziegler1995}. Second, future separation has cost $T_{\mathrm{sep}}$ per query; in general this may contain a hard accumulated-weight reachability problem. Third, after preprocessing, a naive refinement with $N$ reachable states and alphabet size $s$ uses at most $O(N^3T_{\mathrm{sep}}+sN^2)$ work: there are $O(N^2)$ state pairs per refinement round, at most $N-1$ proper split rounds in the naive implementation, and $O(sN^2)$ successor-consistency checks. A polynomial bound is obtained only for model classes with efficient future separation. Fourth, if shortest-path potentials or terminal-offset feasibility are required, the numerical projected residual carrier must still be built or the product graph must be used.

Thus the construction is fully algorithmic in bounded horizons, finite numerical residual quotients, and restricted residual-label models. In the unrestricted deterministic weighted setting it is a characterization plus a separation-based procedure.

\section{Numerical Certificate Layer}

We now recall the exact numerical layer to avoid conflation with the conic abstraction. For a complete extreme-ray representation $\mathcal E(K^\ast)$, cone inclusion is equivalent to nonnegativity of $r(h(w))$ for every extreme ray $r$ and every word $w$. On a finite product graph, a scalar potential certificate for $r$ consists of a map $\Phi_r$ satisfying
\[
        \Phi_r(p')\le \Phi_r(p)+c_r(p,a)
\]
for each edge, together with terminal nonnegativity. This is the standard difference-constraint certificate.

The conic quotient $Q_K(h)$ preserves only signs of residual futures. The numerical carrier for potentials must preserve scalar magnitudes $r(\partial_u h(z))$. Therefore the numerical carrier refines $Q_K(h)$ in general; this is the central separation of the paper.

\begin{procedure}[Fallback exact numerical verification]
Input: a finite reachable weighted carrier $P_R=(P,\Sigma,\delta,g,\tau)$, extreme rays $\mathcal E(K^\ast)$, and the unresolved states or behaviours left after conic screening.

Output: for each extreme ray $r$, either a scalar counterexample path with negative accumulated value after terminal offsets, or a scalar potential certificate proving the corresponding projected difference constraints feasible on the checked product graph.

\begin{enumerate}[label=\arabic*.]
\item Project every vector edge and terminal offset by $r$, obtaining scalar weights $c_r(p,a)=r(g(p,a))$ and scalar terminal offsets $r(\tau(p))$.
\item Keep any conic negative-cell witness as an immediate scalar counterexample for its witnessing ray.
\item On the remaining unresolved part, run the exact scalar difference-constraint procedure, such as Bellman--Ford negative-cycle detection or an equivalent potential computation.
\item Accept the cone inclusion obligation only if every extreme ray returns a valid scalar potential certificate; otherwise return the first scalar counterexample.
\end{enumerate}
\end{procedure}

\begin{proposition}[Soundness of fallback workflow]
If Procedure~1 returns a negative-cell witness, then the corresponding accumulated residual is outside $K$. If Procedure~3 returns scalar potential certificates for every extreme dual ray on the unresolved product graph, then the checked scalarized cone obligations hold exactly on that graph. Conversely, conic nonnegativity alone is insufficient for this conclusion.
\end{proposition}

\begin{proof}
A negative-cell witness has $r(x)<0$ for some $r\in K^\ast$, so $x\notin K$ by definition of the dual cone. For the unresolved part, Procedure~3 is the standard exact scalar difference-constraint check applied separately to each extreme ray. Since a vector lies in $K$ exactly when every extreme dual ray is nonnegative on it, the collection of scalar certificates proves the checked cone obligations. The final statement follows because the conic quotient records only signs, while potentials depend on scalar magnitudes and terminal offsets.
\end{proof}

The relationship can be summarized as the following refinement chain:
\[
\text{vector residual carrier}
\longrightarrow
\text{numerical projected carrier}
\longrightarrow
\text{conic covector carrier}.
\]
The arrows are quotient maps whenever the relevant carriers are finite. Moving right loses information. The first carrier preserves full vector futures; the second preserves the scalar magnitudes seen by each extreme ray and is suitable for potentials; the third preserves only signs and tightness cells. Consequently, conic merging is sound for qualitative obstruction profiles but not for numerical potential feasibility.

\section{Worked Example: Conic Quotient, Witness, and Fallback}

Let $K=K_\triangle$ with extreme dual rays $r^- =(-1,1)$ and $r^+=(1,1)$. Consider a small deterministic carrier whose two residual behaviours are
\[
        x_0=(1,2),
        \qquad
        x_1=(2,4).
\]
They are distinct as vectors and as numerical values. However,
\[
        r^-(x_0)=1,
        \quad r^+(x_0)=3,
        \qquad
        r^-(x_1)=2,
        \quad r^+(x_1)=6.
\]
Both lie in the same open conic cell. If every continuation preserves this proportional relation, the conic quotient has one state while the numerical residual quotient has two.

The full verification workflow on this example has four conceptual phases. First, the conic profile evaluates the two extreme dual rays on each residual vector. Both $x_0$ and $x_1$ have profile $(+,+)$, so the conic quotient merges them. Second, the qualitative screen searches the conic quotient for a reachable profile with a negative coordinate. In the proportional case none appears, so there is no cone-geometric obstruction. Third, numerical verification must still preserve scalar magnitudes such as $r^\pm(x_i)$ and run potential or distance checks if terminal offsets or path sums matter; $x_0$ and $x_1$ may then require different numerical potentials even though their covector type is the same. Finally, if $x_1$ is perturbed to $(4,2)$, then $r^-(x_1)=-2$ and the profile becomes $(-,+)$. This produces a conic negative-cell witness under the ray $r^-$ before any shortest-path potential is computed.

Thus the conic layer can expose a qualitative obstruction when a dual hyperplane is crossed, but it does not certify shortest-path potential feasibility in the merged positive-cell case. The numerical layer remains responsible for scalar magnitudes and terminal offsets. The example also shows why local edge-sign refinement is insufficient: the conic quotient is defined by signs of accumulated residual futures, not by signs of isolated edge increments.

\section{Scope and Limitations}

\subsection*{Determinism, finite words, and facial abstraction}

The optional spectral finite-index note is stated under a dominated conic spectrum assumption. Without such a condition, finite-dimensional linear representations can cross the extreme-dual-ray hyperplanes infinitely often, and the conic quotient may have infinite index. For deterministic finite automata, finite index is guaranteed by finiteness of the reachable presentation.

The results are finite-word and deterministic. Nondeterministic weighted inclusion inherits the classical undecidability and complexity barriers already present in scalar quantitative languages. Infinite-word mean-payoff objectives introduce limit operations that do not generally commute with vector cone observations; they require separate semantics and cannot be obtained by simply taking a limit of the finite-word covector quotient.

For infinite words, a cone condition may be imposed on a vector limit, on componentwise liminf values, or on every scalarized liminf $\liminf_n n^{-1}\sum_{i<n} r(g_i)$. These choices need not agree. Linear functionals commute with genuine Cesaro limits, but they do not generally commute with componentwise liminf or with oscillatory averages near a face of $K$. Consequently, the sign cell of finite prefixes may fail to stabilize, and the conic observation carrier does not automatically yield a mean-payoff certificate. A safe extension would have to fix scalarized mean-payoff semantics first and then rebuild the obstruction algebra for recurrent classes or ultimately periodic witnesses.

For nondeterministic systems, a possible relaxed direction is to use the conic facial algebra as an abstract interpretation domain over sets of cells. This would compute an over-approximation of possible obstruction profiles, prioritize extreme rays that may witness violation, and guide simulation-style inequalities. Such an abstraction would be sound for finding possible qualitative violations but not complete for language inclusion. Candidate tractable fragments include finite-valued or determinizable weighted automata, history-deterministic models, or settings where a quantitative simulation relation supplies a sound proof object. Developing these fragments is left for future work.

\section{Conclusion}

This paper developed cone-induced observation congruences for vector-valued quantitative languages. The construction starts from a rational polyhedral order cone, extracts the extreme-dual covector arrangement, restricts it to the residual span of the language, and forms the right-stable carrier of the resulting residual sign stratification. The carrier construction itself follows classical syntactic-congruence principles; the contribution is the systematic passage from cone geometry to residual observations and the refinement calculus it induces.

The main structural outputs are the conic observation quotient, the row view, the observable-covector correspondence, and the cone-refinement morphisms. These results organize how adding, deleting, or comparing dual covectors changes the induced quotient. The restricted covector configuration may be read as a realizable oriented-matroid sign configuration, while the numerical projected residual carrier remains responsible for magnitude-sensitive potential certificates.

Bounded-horizon quotients are computable by explicit enumeration of accumulated residual futures, and the evaluation section records how the sign abstraction behaves across several deterministic model families. The unbounded construction requires a separation procedure and is best read as a characterization unless such a procedure is available for the model class. The examples illustrate the intended use: conic cells expose qualitative Farkas obstructions and possible collapses before numerical verification is applied.

\section{Appendix: Optional Spectral Finite-Index Regime}

For finite-dimensional linear representations, finite index of $\equiv_K$ is not automatic: the reachable residual orbit may cross extreme-dual hyperplanes infinitely often. A compact diagnostic criterion is obtained by writing
\[
        h(w)=C M_w b,\qquad Y_u=C M_u,
\]
so that
\[
        h(uxz)-h(ux)=Y_uM_x(M_z-I)b.
\]
For each extreme ray $r$ and words $x,z$, the form
\[
        \ell_{r,x,z}(Y)=r\bigl(YM_x(M_z-I)b\bigr)
\]
cuts the row-residual space by a central hyperplane. The conic observation quotient has finite index exactly when the reachable set $\{C M_u:u\in\Sigma^\ast\}$ meets only finitely many cells of this infinite arrangement.

A sufficient spectral regime can be stated under a dominated peripheral/stable decomposition with a positive conic margin away from all nonzero peripheral faces and a separate finite certificate for zero-peripheral faces. In that case long continuations have signs determined by the finite peripheral phase, while short continuations contribute only a bounded prefix factor. The full statement and proof details are included in the ancillary note \texttt{anc/conic-carrier-lab/notes/spectral\_note.md}.

\end{document}